\documentclass[a4paper,USenglish,cleveref, autoref, thm-restate,pdfa]{lipics-v2021}

\usepackage{tikz}
\usetikzlibrary{shapes, positioning}

\pdfoutput=1 
\hideLIPIcs  

\title{Complexity of Finding and Enumerating Interconnection Trees} 

\author{No\'e Demange}{DAVID lab, UVSQ - University Paris-Saclay, France}{noe.demange@uvsq.fr}{https://orcid.org/0009-0007-9438-5552}{}

\author{Yann Strozecki}{DAVID lab, UVSQ - University Paris-Saclay, France}{yann.strozecki@uvsq.fr}{https://orcid.org/0000-0002-0891-3766}{}

\authorrunning{N. Demange and Y. Strozecki} 

\Copyright{No\'e Demange and Yann Strozecki} 



\category{} 

\relatedversion{} 

\supplementdetails[linktext={\\https://github.com/NoeDemange/InterconnectionTreesEnumeration/}, subcategory={Source code}, swhid={swh:1:rev:d63eb0f2de4110e5067a7bebdf5061091fc5e49e}]{Software}{https://github.com/NoeDemange/InterconnectionTreesEnumeration/} 

\acknowledgements{We thank Arnaud Mary for sharing his ideas on the parametrized complexity of the problem.}

\nolinenumbers 

\EventEditors{John Q. Open and Joan R. Access}
\EventNoEds{2}
\EventLongTitle{42nd Conference on Very Important Topics (CVIT 2016)}
\EventShortTitle{CVIT 2016}
\EventAcronym{CVIT}
\EventYear{2016}
\EventDate{December 24--27, 2016}
\EventLocation{Little Whinging, United Kingdom}
\EventLogo{}
\SeriesVolume{42}
\ArticleNo{23}

\usepackage[ruled,vlined]{algorithm2e}
\usepackage{complexity}

\newtheorem*{theorem*}{Theorem}
\newtheorem*{lemma*}{Lemma}

\DeclareMathOperator{\partition}{Part}
\newcommand{\ITP}{\textsc{ITP}}
\newcommand{\CITP}{\textsc{\#ITP}}
\newcommand{\EITP}{\textsc{EnumITP}}

\DeclareMathOperator{\IT}{IT}
\newcommand{\ITof}[1]{\IT\!\left(#1\right)}

\newcommand{\T}{\rule{0pt}{2ex}}

\newcommand{\mcomment}[2]{\ifmmode\margincomment{#1}{#2}\else\footcomment{#1}{#2}\fi}
\newcommand{\footcomment}[2]{{\color{blue}\textbf{(#1)}}\footnote{\textbf{#1:} #2}}
\newcommand{\margincomment}[2]{{\color{blue}\textbf{(#1)}}\footnotemark\marginnote{\tiny\textsuperscript{\thefootnote}\textbf{#1:} #2}}

\begin{document}

\maketitle

\begin{abstract}
We study the problem of connecting the parts of a multipartite graph using a minimum number of edges under a matching constraint. We introduce \emph{interconnection trees}, defined as matchings whose projections onto the quotient graph form a spanning tree. Motivated by applications in chemoinformatics, we investigate the decision, counting, and enumeration variants of this problem.

We show that the decision problem is $\NP$-complete. Nevertheless, it becomes tractable in several structured settings: it is fixed-parameter tractable in the number of parts, and admits polynomial or linear-time algorithms on complete, quasi-complete, and $t$-quasi-complete multipartite graphs. We also study enumeration, for which we design efficient flashlight-search based algorithms with optimal delay for complete multipartite graphs, and a weight-guided heuristic that prioritizes low-weight solutions and performs well in practice.
\end{abstract}

\newpage
\section{Introduction}
\label{sec:Introduction}
We study a graph-theoretic problem arising from a practical question in computational chemistry~\cite{bricage2018,demange2026computational}. In the design of molecular cages, one first places several binding patterns --- small molecular fragments positioned to interact with a target molecule. These patterns must then be connected by molecular paths into a single covalent structure, without reusing atoms and without introducing redundant connections. Among all feasible constructions, those minimizing the total length of the molecular path are preferred, as this quantity correlates with both the physical size and the synthetic feasibility of the cage.

We model this setting using a multipartite graph whose parts represent binding patterns, whose vertices represent candidate junction atoms, and whose edges represent possible connections by a molecular path between atoms in distinct patterns. We abstract away the construction of the molecular path, studied in detail in~\cite{demange2026computational} and focus only on the first step of deciding how to connect all parts using a set of vertex-disjoint edges that induces a spanning tree at the level of parts. We call such a set an \emph{interconnection tree}: a matching whose projection onto the quotient graph of parts is a spanning tree. 

Although motivated by cheminformatics, this abstraction captures a broader class of connection problems under compatibility and resource constraints. For instance, in network design, servers may expose multiple communication ports, but only certain pairs of ports are mutually compatible. Selecting a set of disjoint links that connects all servers corresponds exactly to finding an interconnection tree.

From an algorithmic perspective, two aspects make this problem nonstandard. First, it combines two classical structures --- matchings and spanning trees --- under a coupling constraint induced by the multipartite structure. Second, in the motivating applications, computing a single feasible solution is often insufficient. Geometric or physical constraints may invalidate a candidate construction, and one must be able to explore many alternatives. This naturally leads to the study of the problem of enumerating all interconnection trees.

In enumeration problems, the relevant complexity measure is not only the total running time, but also the worst-case \emph{delay} between two consecutive outputs. When all solutions can be generated, one also considers the \emph{average delay}, defined as the total running time divided by the number of solutions (see~\cite{johnson1988generating,phd} for more details on enumeration complexity). Classical techniques achieving polynomial or even constant amortized delay include Gray codes, reverse search~\cite{avis1996reverse}, and the flashlight-search technique~\cite{boros2009generating,capelli2021enumerating}. The latter provides a general framework for enumeration based on recursive partitioning of the solution space combined with efficient feasibility tests.

A further challenge arises from edge weights: in practice, one is primarily interested in low-weight solutions. Generating combinatorial objects in nondecreasing weight order is substantially harder than unordered enumeration, and for many problems such ordered generation remains difficult; see~\cite{eppstein2015k,phd} and the references therein.

\subparagraph*{Related works.}
Complete multipartite graphs frequently appear in graph theory as extremal examples, such as Turán graphs (balanced complete multipartite graphs), which characterize edge-maximal graphs without large cliques~\cite{turan1941}. Their symmetry also enables closed-form counting of otherwise difficult combinatorial structures, such as spanning trees~\cite{moon1967}.

Spanning tree and matching problems are classical problems solvable in polynomial time. Minimum spanning trees can be computed in near-linear time via Kruskal’s or Prim’s algorithm~\cite{kruskal1956shortest,prim1957shortest}, and maximum matchings in general graphs can be found in polynomial time using Edmonds’ blossom algorithm~\cite{edmonds1965paths}. However, many variants of these problems, such as the Steiner tree problem~\cite{Karp1972}, degree-constrained spanning tree or matching~\cite{garey1979computers} are $\NP$-hard.

Enumeration of spanning trees has a long history, with algorithms achieving polynomial delay~\cite{gabow1978finding} even in increasing weight order, constant amortized time per solution~\cite{kapoor1995}, and Gray code orderings~\cite{merino2022all}. Enumerating matchings can also be achieved with constant amortized delay~\cite{uno2015constant}, and related techniques will be used in this work. Moreover, weighted matchings in bipartite graphs admit ranking algorithms such as Murty’s method~\cite{murty1968algorithm}, which produce solutions in increasing order of cost.

\subparagraph*{Main contributions.}
In \cref{ITProb}, we show that \ITP, the problem of deciding whether an interconnection tree exists, is $\NP$-complete. Despite this hardness, we prove that \ITP\ is fixed-parameter tractable with respect to the number of parts. We then identify several tractable subclasses. For complete multipartite graphs, we derive a structural characterization together with a closed-form formula for the number of interconnection trees. We extend these results to quasi-complete and $t$-quasi-complete multipartite graphs, where the problem is in $\P$ or in $\XP$ when parameterized by $t$.

In \cref{ICTEnum}, we study \EITP, the problem of enumerating all interconnection trees. We design flashlight-search algorithms and show that one of them achieves optimal delay and amortized delay. Finally, in \cref{WICTEnum}, we introduce a weight-guided enumeration heuristic that tends to produce low-weight solutions early, and we evaluate its practical behavior on synthetic and real molecular instances.

\section{Definitions}
\label{sec:Definitions}

A graph $G=(V,E)$ is \textbf{multipartite} if its vertex set can be partitioned into
parts $V_1,\dots,V_k$ such that each $V_i$ is non-empty and induces an independent set of $G$.
Let $k$ be the number of parts in a multipartite graph and let $[k]=\{1,\dots,k\}$. We define the function $\partition : V \to [k]$ by $\partition(v)=i$ if and only if $v\in V_i$. Whenever we consider a multipartite graph, we assume that it is given together with the function $\partition$.

The chemical setting we model motivates multipartite graphs in which all edges between distinct parts are present. A \textbf{complete multipartite graph} is a multipartite graph
$G=(V=V_1\sqcup\dots\sqcup V_k,E)$ such that for all $i\neq j$ in $[k]$, if $u\in V_i$ and $v\in V_j$, then $(u,v)\in E$.
We relax this notion slightly by not enforcing completeness for the first $t$ parts.

\begin{definition}
A \textbf{$t$-quasi-complete multipartite graph} is a multipartite graph
$G=(V=V_1\sqcup\dots\sqcup V_k,E)$ such that for all $i\neq j$,
if $i>t$ and $j>t$, then for all $u\in V_i$ and $v\in V_j$, we have $(u,v)\in E$.
\end{definition}

The most common cases considered in this article are $t=0$, corresponding to complete multipartite graphs, and $t=1$, for which we say that the graph is \textbf{quasi-complete}. In a $t$-quasi-complete multipartite graph, we call the parts $V_i$ with $i>t$ the \textbf{complete} parts.

We extend the function $\partition$ to edges and sets of edges by defining
$\partition((u,v)) = (\partition(u),\partition(v))$
and
$ \partition(E)=\{\partition(e)\mid e\in E\}.$
To represent the connections between the parts of a multipartite graph, we introduce the quotient graph, as illustrated in \cref{fig:sol_motLiant_ACANIL}.

\begin{figure}[t]
    \centering
    \begin{subfigure}[t]{0.49\textwidth}
    \centering
    \resizebox{0.5\linewidth}{!}{
        \begin{tikzpicture}[scale=0.8]
        \node[draw, circle, color=blue, minimum size=2cm] (C1) at (0,4) {};
        \node[anchor=center, color=blue, color=blue] (C1label) at (0, 3.2) {$V_1$};
        \node[draw, circle, minimum size=0.8cm, inner sep=0pt] (1) at (-0.5, 4) {1};
        \node[draw, circle, minimum size=0.8cm, inner sep=0pt] (2) at (0.5, 4) {2};
    
        \node[draw, circle, color=blue, minimum size=2cm] (C2) at (5,4) {};
        \node[anchor=center, color=blue] (C2label) at (5, 3.1) {$V_2$};
        \node[draw, circle, minimum size=0.8cm, inner sep=0pt] (3) at (4.5, 4) {3};
        \node[draw, circle, minimum size=0.8cm, inner sep=0pt] (4) at (5.5, 4) {4};
    
        \node[draw, circle, color=blue, minimum size=3cm] (C3) at (0,0) {};
        \node[anchor=center, color=blue] (C3label) at (0, -1.6) {$V_3$};
        \node[draw, circle, minimum size=0.8cm, inner sep=0pt] (5) at (-0.5, 0) {5};
        \node[draw, circle, minimum size=0.8cm, inner sep=0pt] (6) at (-0.5, 1) {6};
        \node[draw, circle, minimum size=0.8cm, inner sep=0pt] (7) at (0.5, 1) {7};
        \node[draw, circle, minimum size=0.8cm, inner sep=0pt] (8) at (0.5, 0) {8};
        \node[draw, circle, minimum size=0.8cm, inner sep=0pt] (9) at (0.5, -1) {9};
        \node[draw, circle, minimum size=0.8cm, inner sep=0pt] (10) at (-0.5, -1) {10};
    
        \node[draw, circle, color=blue, minimum size=3cm] (C4) at (5,0) {};
        \node[anchor=center, color=blue] (C4label) at (5, -1.6) {$V_4$};
        \node[draw, circle, minimum size=0.8cm, inner sep=0pt] (11) at (4.5, 0) {11};
        \node[draw, circle, minimum size=0.8cm, inner sep=0pt] (12) at (4.5, 1) {12};
        \node[draw, circle, minimum size=0.8cm, inner sep=0pt] (13) at (5.5, 1) {13};
        \node[draw, circle, minimum size=0.8cm, inner sep=0pt] (14) at (5.5, 0) {14};
        \node[draw, circle, minimum size=0.8cm, inner sep=0pt] (15) at (5.5, -1) {15};
        \node[draw, circle, minimum size=0.8cm, inner sep=0pt] (16) at (4.5, -1) {16};

        \draw[color = red] (1) -- (6);
        \draw[color = black] (1) -- (7);
        \draw[color = black, bend right=80] (1) edge (5);
        \draw[color = black, bend left=80] (1) edge (4);
        \draw[color = black, bend left=80] (1) edge (3);
        \draw[color = red] (2) -- (3);
        \draw[color = black] (3) -- (7);
        \draw[color = black] (3) -- (8);
        \draw[color = black] (3) -- (12);
        \draw[color = black] (4) -- (8);
        \draw[color = black] (4) -- (13);
        \draw[color = black, bend left=80] (4) edge (14);
        \draw[color = red] (7) -- (11);
        \draw[color = black] (8) -- (12);
        \draw[color = black] (9) -- (16);
        \draw[color = black, bend right=80] (10) edge (15);

        \end{tikzpicture}
        }
        \label{fig:sol1_inter_motLiant_ACANIL}
    \end{subfigure}
    \begin{subfigure}[t]{0.49\textwidth}
    \centering
    \resizebox{0.5\linewidth}{!}{
        \begin{tikzpicture}[scale=0.8]
        \node[draw, circle, color=blue, minimum size=2cm] (V1) at (0,4) {$V_1$};
        \node[draw, circle, color=blue, minimum size=2cm] (V2) at (5,4) {$V_2$};

        \node[draw, circle, color=blue, minimum size=3cm] (V3) at (0,0) {$V_3$};
        \node[draw, circle, color=blue, minimum size=3cm] (V4) at (5,0) {$V_4$};
       
        \draw[color = red] (V1) -- (V3);
        \draw[color = red] (V1) -- (V2);
        \draw[color = red] (V3) -- (V4);
        \draw[color = black] (V2) -- (V4);
        \draw[color = black] (V2) -- (V3);

        \end{tikzpicture}
        }
        \label{fig:Gpart_with_IT}
    \end{subfigure}
    \caption{A multipartite graph $G$ (left) and the corresponding quotient graph $G_{\partition}$ (right). In red, edges of an interconnection tree $T = \{(1,6), (2,3), (7,11)\}$. }
    \label{fig:sol_motLiant_ACANIL}
\end{figure}
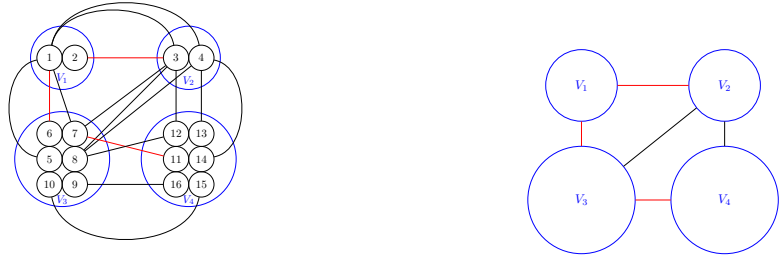

\begin{definition}
Given a multipartite graph $G=(V=V_1\sqcup\dots\sqcup V_k,E)$,
the \textbf{quotient graph} is defined as
$ G_{\partition}=([k],\partition(E)).$
\end{definition}

We are interested in connecting all parts of a multipartite graph
using a minimum number of edges without reusing vertices.
This motivates the notion of an interconnection tree, which is the main object studied in this article.

\begin{definition}
An \textbf{interconnection tree} of a multipartite graph
$G=(V=V_1\sqcup\dots\sqcup V_k,E)$ is a set $T\subseteq E$ such that
\begin{enumerate}
    \item $T$ is a matching in $G$;
    \item $\partition(T)$ is a spanning tree of $G_{\partition}$.
\end{enumerate}
\end{definition}

We denote by $\ITof{G}$ the set of interconnection trees of a multipartite graph $G$. Given $G$, we consider the problem of listing all elements of $\ITof{G}$, which we call \EITP\ for Enumeration of Interconnection Trees Problem. Efficient enumeration of these objects first requires solving the associated decision problem \ITP, which asks, given $G$, whether $\ITof{G}$ is empty. We also consider the counting variant of the problem, denoted by \CITP, which consists in computing $|\ITof{G}|$.

\section{Interconnection Tree Problem}
\label{ITProb}
In this section, we study the complexity of \ITP, first for general multipartite graphs. We then focus on complete and quasi-complete multipartite graphs, which are relevant for applications and for which \ITP\ is tractable.

\subsection{Complexity of \ITP}

Like many variants of the spanning tree problem, \ITP\ is hard by reduction from the Hamiltonian Path problem.

\begin{theorem} \label{th:ITisNP}
\ITP\ is $\NP$-complete.
\end{theorem}

\begin{proof}
 Given a set $T \subseteq E$, we can check in linear time
that $T$ is a matching of $G$ and that $\partition(T)$ is a spanning tree of $G_{\partition}$. Hence, $\ITP \in \NP$.

We prove $\NP$-hardness by reduction from the Hamiltonian Path problem~\cite{DBLP:conf/coco/Karp72}.
Let $H$ be a graph, instance of Hamiltonian Path. We construct a multipartite graph $G$ as follows. For each $u \in V(H)$, define a part
$V_u = \{u^{\text{in}}, u^{\text{out}}\}$.
Let
\(
V = \bigsqcup_{u \in V(H)} V_u \)
and
\(
E = \{(u^{\text{out}}, v^{\text{in}}) \mid (u,v) \in E(H)\}.
\)
The parts of $G$ are the sets $V_u$. This construction is linear in the size of $H$ and is illustrated in \cref{fig:NP_constrution}.

Assume that $H$ has a Hamiltonian path
$P = u_1, u_2, \dots, u_n$. Define $
T = \{(u_i^{\text{out}}, u_{i+1}^{\text{in}}) \mid 1 \le i \le n-1\}$.
Since $P$ visits each vertex exactly once, the edges of $T$ have pairwise distinct
endpoints, hence $T$ is a matching of $G$.
Moreover, each edge of $\partition(T)$ is of the form
$(\partition(u_i^{\text{out}}), \partition(u_{i+1}^{\text{in}}))$,
so $\partition(T)$ is a path in the quotient graph $G_{\partition}$.
Because $P$ is Hamiltonian, this path visits all vertices of $G_{\partition}$.
Hence, $\partition(T)$ is a spanning tree of $G_{\partition}$,
and $T$ is an interconnection tree of $G$.

Conversely, assume that $G$ has an interconnection tree $T$.
Then $T$ is a matching such that $\partition(T)$ is a spanning tree of $G_{\partition}$.
Each part $V_u$ contains exactly two vertices, so at most two edges of $T$
can be incident to vertices of the same part.
It follows that every vertex of $\partition(T)$ has degree at most two.
Since $\partition(T)$ is a spanning tree, it must be a path.

We can order the edges of $T$ in path order as
$(u_i^{\text{out}}, u_{i+1}^{\text{in}})$ for $i=1,\dots,n-1$.
By construction of $G$, if
$(u_i^{\text{out}}, u_{i+1}^{\text{in}}) \in E(G)$,
then $(u_i, u_{i+1}) \in E(H)$.
Hence, $u_1, \dots, u_n$ is a Hamiltonian path of $H$.
\end{proof}

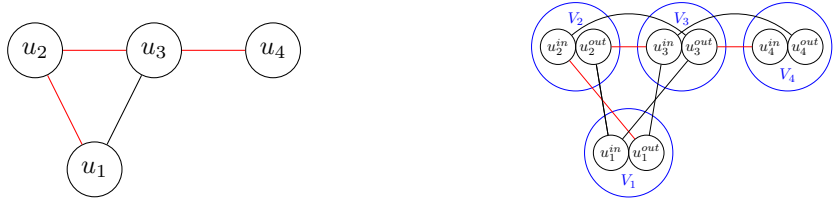
\begin{figure}[t]
    \centering
    \begin{subfigure}[t]{0.49\textwidth}
    \centering
    \resizebox{0.6\linewidth}{!}{
        \begin{tikzpicture}[scale=0.8]
        \node[circle,draw] (a) at (1,-2) {$u_1$};
        \node[circle,draw] (b) at (0,0) {$u_2$};
        \node[circle,draw] (c) at (2,0) {$u_3$};
        \node[circle,draw] (d) at (4,0) {$u_4$};
        \draw[color=red] (a)--(b);
        \draw (a)--(c);
        \draw[color=red] (b)--(c);
        \draw[color=red] (c)--(d);
        \end{tikzpicture}
    }
         \label{fig:NP_H}
    \end{subfigure}
    \begin{subfigure}[t]{0.49\textwidth}
    \centering
    \resizebox{0.6\linewidth}{!}{
            \begin{tikzpicture}[scale=0.8]
            \node[draw, circle, color=blue, minimum size=2cm] (C1) at (1.5,-3) {};
            \node[anchor=center, color=blue, color=blue] (C1label) at (1.5, -3.8) {$V_1$};
            \node[draw, circle, minimum size=0.8cm, inner sep=0pt] (1) at (1, -3) {$u_1^{in}$};
            \node[draw, circle, minimum size=0.8cm, inner sep=0pt] (2) at (2, -3) {$u_1^{out}$};
            \node[draw, circle, color=blue, minimum size=2cm] (C2) at (0,0) {};
            \node[anchor=center, color=blue, color=blue] (C2label) at (0, 0.8) {$V_2$};
            \node[draw, circle, minimum size=0.8cm, inner sep=0pt] (3) at (-0.5, 0) {$u_2^{in}$};
            \node[draw, circle, minimum size=0.8cm, inner sep=0pt] (4) at (0.5, 0) {$u_2^{out}$};
        
            \node[draw, circle, color=blue, minimum size=2cm] (C3) at (3,0) {};
            \node[anchor=center, color=blue] (C3label) at (3, 0.8) {$V_3$};
            \node[draw, circle, minimum size=0.8cm, inner sep=0pt] (5) at (2.5, 0) {$u_3^{in}$};
            \node[draw, circle, minimum size=0.8cm, inner sep=0pt] (6) at (3.5, 0) {$u_3^{out}$};
        
            \node[draw, circle, color=blue, minimum size=2cm] (C4) at (6,0) {};
            \node[anchor=center, color=blue] (C4label) at (6, -0.8) {$V_4$};
            \node[draw, circle, minimum size=0.8cm, inner sep=0pt] (7) at (5.5, 0) {$u_4^{in}$};
            \node[draw, circle, minimum size=0.8cm, inner sep=0pt] (8) at (6.5, 0) {$u_4^{out}$};
    
            \draw[color = red] (4) -- (5);
            \draw[color = red] (6) -- (7);
            \draw (1) -- (4);
            \draw[color = red] (2) -- (3);
            \draw (1) -- (6);
            \draw (2) -- (5);
            \draw (1) -- (4);
            \draw[bend left=40] (3) edge (6);
            \draw[bend left=40] (5) edge (8);
           
            \end{tikzpicture}
        }
        \label{fig:NP_G}
    \end{subfigure}
    \caption{Construction of the graph $G$ (right) from a graph $H$ (left). In red, the Hamiltonian path and the interconnection tree.}
    \label{fig:NP_constrution}
\end{figure}

The reduction given in the proof of \cref{th:ITisNP} also proves that \CITP\ is $\#\P$-complete, since it gives a bijection between the interconnection trees and the Hamiltonian paths, and counting Hamiltonian paths is $\#\P$-hard~\cite{valiant1979complexity}.

Since \ITP\ is $\NP$-hard, we now study its parameterized complexity. A natural parameter is the number of parts $k$. We show that the problem is fixed-parameter tractable with respect to $k$, by considering all possible spanning trees of the reduced graph and applying a kernelization to the multipartite graph for each of these trees. 

\begin{theorem}
\ITP\ can be solved in time
$O(m\sqrt{n} + 2^k k^{3k})$ on a multipartite graph with $n$ vertices,
$m$ edges, and $k$ parts.
\end{theorem}

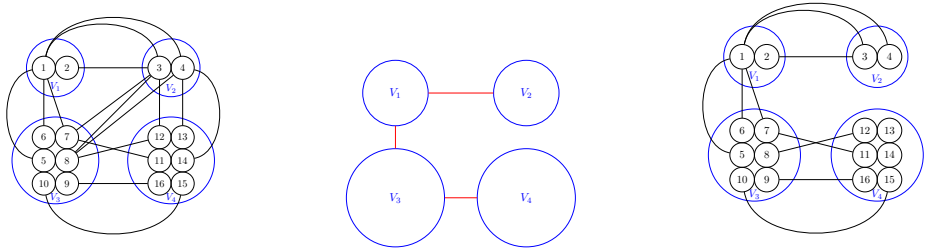
\begin{figure}[b]
    \centering
    \begin{subfigure}[t]{0.32\textwidth}
    \centering
    \resizebox{0.7\linewidth}{!}{
        \begin{tikzpicture}[scale=0.8]
        \node[draw, circle, color=blue, minimum size=2cm] (C1) at (0,4) {};
        \node[anchor=center, color=blue, color=blue] (C1label) at (0, 3.2) {$V_1$};
        \node[draw, circle, minimum size=0.8cm, inner sep=0pt] (1) at (-0.5, 4) {1};
        \node[draw, circle, minimum size=0.8cm, inner sep=0pt] (2) at (0.5, 4) {2};
    
        \node[draw, circle, color=blue, minimum size=2cm] (C2) at (5,4) {};
        \node[anchor=center, color=blue] (C2label) at (5, 3.1) {$V_2$};
        \node[draw, circle, minimum size=0.8cm, inner sep=0pt] (3) at (4.5, 4) {3};
        \node[draw, circle, minimum size=0.8cm, inner sep=0pt] (4) at (5.5, 4) {4};
    
        \node[draw, circle, color=blue, minimum size=3cm] (C3) at (0,0) {};
        \node[anchor=center, color=blue] (C3label) at (0, -1.6) {$V_3$};
        \node[draw, circle, minimum size=0.8cm, inner sep=0pt] (5) at (-0.5, 0) {5};
        \node[draw, circle, minimum size=0.8cm, inner sep=0pt] (6) at (-0.5, 1) {6};
        \node[draw, circle, minimum size=0.8cm, inner sep=0pt] (7) at (0.5, 1) {7};
        \node[draw, circle, minimum size=0.8cm, inner sep=0pt] (8) at (0.5, 0) {8};
        \node[draw, circle, minimum size=0.8cm, inner sep=0pt] (9) at (0.5, -1) {9};
        \node[draw, circle, minimum size=0.8cm, inner sep=0pt] (10) at (-0.5, -1) {10};
    
        \node[draw, circle, color=blue, minimum size=3cm] (C4) at (5,0) {};
        \node[anchor=center, color=blue] (C4label) at (5, -1.6) {$V_4$};
        \node[draw, circle, minimum size=0.8cm, inner sep=0pt] (11) at (4.5, 0) {11};
        \node[draw, circle, minimum size=0.8cm, inner sep=0pt] (12) at (4.5, 1) {12};
        \node[draw, circle, minimum size=0.8cm, inner sep=0pt] (13) at (5.5, 1) {13};
        \node[draw, circle, minimum size=0.8cm, inner sep=0pt] (14) at (5.5, 0) {14};
        \node[draw, circle, minimum size=0.8cm, inner sep=0pt] (15) at (5.5, -1) {15};
        \node[draw, circle, minimum size=0.8cm, inner sep=0pt] (16) at (4.5, -1) {16};

        \draw[color = black] (1) -- (6);
        \draw[color = black] (1) -- (7);
        \draw[color = black, bend right=80] (1) edge (5);
        \draw[color = black, bend left=80] (1) edge (4);
        \draw[color = black, bend left=80] (1) edge (3);
        \draw[color = black] (2) -- (3);
        \draw[color = black] (3) -- (7);
        \draw[color = black] (3) -- (8);
        \draw[color = black] (3) -- (12);
        \draw[color = black] (4) -- (8);
        \draw[color = black] (4) -- (13);
        \draw[color = black, bend left=80] (4) edge (14);
        \draw[color = black] (7) -- (11);
        \draw[color = black] (8) -- (12);
        \draw[color = black] (9) -- (16);
        \draw[color = black, bend right=80] (10) edge (15);

        \end{tikzpicture}
        }
        \label{fig:FPT_G}
    \end{subfigure}
    \begin{subfigure}[t]{0.32\textwidth}
    \centering
    \resizebox{0.7\linewidth}{!}{
        \begin{tikzpicture}[scale=0.8]
        \node[draw, circle, color=blue, minimum size=2cm] (V1) at (0,4) {$V_1$};
        \node[draw, circle, color=blue, minimum size=2cm] (V2) at (5,4) {$V_2$};

        \node[draw, circle, color=blue, minimum size=3cm] (V3) at (0,0) {$V_3$};
        \node[draw, circle, color=blue, minimum size=3cm] (V4) at (5,0) {$V_4$};
       
        \draw[color = red] (V1) -- (V3);
        \draw[color = red] (V1) -- (V2);
        \draw[color = red] (V3) -- (V4);

        \end{tikzpicture}
        }
        \label{fig:FPT_S}
    \end{subfigure}
    \begin{subfigure}[t]{0.32\textwidth}
    \centering
    \resizebox{0.7\linewidth}{!}{
        \begin{tikzpicture}[scale=0.8]
        \node[draw, circle, color=blue, minimum size=2cm] (C1) at (0,4) {};
        \node[anchor=center, color=blue, color=blue] (C1label) at (0, 3.2) {$V_1$};
        \node[draw, circle, minimum size=0.8cm, inner sep=0pt] (1) at (-0.5, 4) {1};
        \node[draw, circle, minimum size=0.8cm, inner sep=0pt] (2) at (0.5, 4) {2};
    
        \node[draw, circle, color=blue, minimum size=2cm] (C2) at (5,4) {};
        \node[anchor=center, color=blue] (C2label) at (5, 3.1) {$V_2$};
        \node[draw, circle, minimum size=0.8cm, inner sep=0pt] (3) at (4.5, 4) {3};
        \node[draw, circle, minimum size=0.8cm, inner sep=0pt] (4) at (5.5, 4) {4};
    
        \node[draw, circle, color=blue, minimum size=3cm] (C3) at (0,0) {};
        \node[anchor=center, color=blue] (C3label) at (0, -1.6) {$V_3$};
        \node[draw, circle, minimum size=0.8cm, inner sep=0pt] (5) at (-0.5, 0) {5};
        \node[draw, circle, minimum size=0.8cm, inner sep=0pt] (6) at (-0.5, 1) {6};
        \node[draw, circle, minimum size=0.8cm, inner sep=0pt] (7) at (0.5, 1) {7};
        \node[draw, circle, minimum size=0.8cm, inner sep=0pt] (8) at (0.5, 0) {8};
        \node[draw, circle, minimum size=0.8cm, inner sep=0pt] (9) at (0.5, -1) {9};
        \node[draw, circle, minimum size=0.8cm, inner sep=0pt] (10) at (-0.5, -1) {10};
    
        \node[draw, circle, color=blue, minimum size=3cm] (C4) at (5,0) {};
        \node[anchor=center, color=blue] (C4label) at (5, -1.6) {$V_4$};
        \node[draw, circle, minimum size=0.8cm, inner sep=0pt] (11) at (4.5, 0) {11};
        \node[draw, circle, minimum size=0.8cm, inner sep=0pt] (12) at (4.5, 1) {12};
        \node[draw, circle, minimum size=0.8cm, inner sep=0pt] (13) at (5.5, 1) {13};
        \node[draw, circle, minimum size=0.8cm, inner sep=0pt] (14) at (5.5, 0) {14};
        \node[draw, circle, minimum size=0.8cm, inner sep=0pt] (15) at (5.5, -1) {15};
        \node[draw, circle, minimum size=0.8cm, inner sep=0pt] (16) at (4.5, -1) {16};

        \draw[color = black] (1) -- (6);
        \draw[color = black] (1) -- (7);
        \draw[color = black, bend right=80] (1) edge (5);
        \draw[color = black, bend left=80] (1) edge (4);
        \draw[color = black, bend left=80] (1) edge (3);
        \draw[color = black] (2) -- (3);
        \draw[color = black] (7) -- (11);
        \draw[color = black] (8) -- (12);
        \draw[color = black] (9) -- (16);
        \draw[color = black, bend right=80] (10) edge (15);

        \end{tikzpicture}
        }
        \label{fig:FPT_GexcludeProj(e)notinS}
    \end{subfigure}
    \caption{A multipartite graph $G$ (left), a spanning tree $S$ on the graph of parts $G_{\mathrm{part}}$ (middle) and $G_S$ (right).}
    \label{fig:FPT_random_multipartite}
\end{figure}

\begin{proof}
Let $G=(V=V_1\sqcup\dots\sqcup V_k,E)$ be a multipartite graph.
We assume it is connected otherwise, no interconnection tree exists.

Let $S$ be a spanning tree of the quotient graph $G_{\partition}$.
We first describe an algorithm that decides whether there is
$T\in\ITof{G}$ such that $\partition(T)=S$.  
Define $G_S=(V,E_S)$ by keeping only the edges whose projection lies in $S$:
\(
E_S=\{e\in E \mid \partition(e)\in S\}
\)
(Figure~\ref{fig:FPT_random_multipartite}).
Then $T$ is an interconnection tree of $G$ such that $\partition(T)=S$
if and only if $T\in\ITof{G_S}$.

We show how to decompose the problem along an edge of $S$. W.l.o.g. let $e=(1,2)\in S$, and let $G_{S,e}$ be the bipartite graph of $G_S$
induced by $V_1\sqcup V_2$.
The graph $S\setminus\{e\}$ has two connected components,
denoted $S_1$ and $S_2$.
Removing the edges between $V_1$ and $V_2$ in $G_S$
splits it into two connected multipartite graphs
$G_1$ and $G_2$, whose parts are connected by $S_1$ and $S_2$ respectively.

\begin{claim}\label{claim:dec}
If $G_{S,e}$ admits a matching of size $k-1$, then
there exists $T\in\ITof{G}$ with $\partition(T)=S$
if and only if there exist $T_1\in\ITof{G_1}$ 
and $T_2\in\ITof{G_2}$ such that $\partition(T_1)=S_1$ and $\partition(T_2)=S_2$.
\end{claim}

\begin{proof}
Assume first that $T\in\ITof{G}$ with $\partition(T)=S$.
Let $T_1$ and $T_2$ be the subsets of $T$ whose projections are
$S_1$ and $S_2$, respectively. They are matchings because they are subsets of a matching. Moreover, $S_1$ and $S_2$ are spanning trees of the quotient graphs of $G_1$ and $G_2$, we have $T_1\in\ITof{G_1}$ and $T_2\in\ITof{G_2}$.

Conversely, let $M$ be a matching of $G_{S,e}$ of size $k-1$,
and assume that $T_1 \in \ITof{G_1}$ and $T_2\in\ITof{G_2}$ with $\partition(T_1)=S_1$ and $\partition(T_2)=S_2$. Let $U_1\subseteq V_1$ and $U_2\subseteq V_2$ be the vertices incident to edges of $T_1\cup T_2$.
Let $k_1$ and $k_2$ be the numbers of parts in $G_1$ and $G_2$,
so $k=k_1+k_2$.
Since $|T_1|=k_1-1$ and $|T_2|=k_2-1$,
we have $|U_1|+|U_2|\le (k_1-1)+(k_2-1)=k-2$.
By the pigeonhole principle,
some edge $e'\in M$ is incident to no vertex of $U_1\cup U_2$.
Then $T=T_1\cup T_2\cup\{e'\}$ is a matching,
and $\partition(T)=S$.
\end{proof}

By \cref{claim:dec}, we may assume that for every edge $(i,j)\in S$,
the maximum matching of the bipartite graph induced by $V_i$ and $V_j$
has size at most $k-2$, otherwise we can solve the problem by reduction to two instances, whose number of parts sum to $k$.

For each $(i,j)\in S$, let $G_{i,j}$ be this bipartite graph, and let $C_{i,j}$ be a minimum vertex cover obtained from a maximum matching.
By K\H{o}nig's theorem~\cite{DBLP:books/others/BondyM76}, $|C_{i,j}|<k-1$.
We now reduce $G_S$ by removing vertices that are irrelevant for the
existence of an interconnection tree. Let $d_i$ be the degree of $i$ in $S$.
For $(i,j)\in S$ and each $c\in C_{i,j}\cap V_j$,
consider the neighbors of $c$ in $V_i$ that are not in $C_{i,j}$.
If there are more than $d_i$ such vertices, keep only $d_i$ of them.
Applying this rule to all pairs $(i,j)$ yields a reduced graph $G'_S$.

Let $T \in \ITof{G_S}$. Since $C$ is a vertex cover of $G_S$, every edge of $T$ has an endpoint in $C$. Hence, vertices of $V\setminus C$ appear in $T$ only through edges whose other endpoint lies in $C$.

Consider $(i,j)\in S$ and the edge $(v_i,v_j)\in T$ with
$v_j \in C_{i,j}$. If $v_i$ was removed when constructing $G'_S$,
then by construction there are at least $d_i$ vertices in $V_i$
still present in $G'_S$ and adjacent to $v_j$.
Since exactly $d_i -1$ vertices of $V_i \setminus \{v_i\} $ are incident to edges of $T$, one of the neighbors of $v_j$ kept in $G'_S$, say $u$, is not used in $T$.
Replacing $(v_i,v_j)$ by $(u,v_j)$ preserves both the matching property
and the projection onto $(i,j)$.

Applying this replacement to every removed vertex used by $T$
yields an interconnection tree contained in $G'_S$. We have thus proved that 
 if $\ITof{G_S}$ is non-empty, then so is $\ITof{G'_S}$.

Now, the size of $G'_S$ depends only on $k$.
We bound the number of edge sets that must be considered to decide
whether there exists an interconnection tree $T$ with $\partition(T)=S$.

For each edge $(i,j)\in S$, we must select one edge between $V_i$ and $V_j$.
Since $|C_{i,j}|\le k-1$, there are at most $(k-1)^2$ edges
with both endpoints in $C_{i,j}$.
Moreover, in $G'_S$, each vertex of $C_{i,j}\cap V_j$
is adjacent to at most $d_i$ vertices of $V_i\setminus C_{i,j}$.
As $d_i\le k-1$, this yields at most $(k-1)^2$ edges
with exactly one endpoint in $C_{i,j}$.
Hence, for each pair $(i,j)\in S$,
there are at most $2(k-1)^2$ candidate edges.

Therefore, we need to consider at most $ (2(k-1)^2)^{k-1}$
sets of $k-1$ edges. These sets can be generated with amortized constant time per set, and testing whether a set forms a matching can also be done
with constant amortized overhead.

We assume that, before choosing $S$,
maximum matchings in all graphs $G_{i,j}$ are computed,
the corresponding vertex covers $C_{i,j}$ are derived,
and that all vertices to remove to obtain $G'_S$ are precomputed.
This preprocessing can be done in time $O(m\sqrt{n})$
using the Hopcroft--Karp algorithm~\cite{DBLP:conf/focs/HopcroftK71}.

To solve \ITP, we solve the problem of deciding whether there is $T\in\ITof{G}$ such that $\partition(T)=S$ for every spanning tree $S$
of $G_{\partition}$. The number of such trees is at most $k^{k-2}$,
and they can be enumerated with amortized constant delay~\cite{Aigner2018,kapoor1995algorithms,shioura1997optimal}.
Since the preprocessing is done only once,
the total running time is in $O(m\sqrt{n} + 2^k k^{3k})$.
\end{proof}


\subsection{Complete multipartite graphs}

We adapt the classical edge contraction operation to multipartite graphs as follows.

\begin{definition}\label{definition:contraction}
Let $G=(V=V_1\sqcup\dots\sqcup V_k,E)$ be a multipartite graph, and let $e=(u,v)\in E$ with $u\in V_i$ and $v\in V_j$. If $|V_i| \neq 1$ or 
$|V_j| \neq 1$, the result of the contraction of $e$ is the graph $G_{/e}=(V \setminus \{u,v\},E \setminus \{(w,z)\in E \mid w\in V_i \text{ and } z\in V_j\})$.
The graph $G_{/e}$ inherits the same partition as $G$,
except that the parts $V_i$ and $V_j$ are replaced by the single part $(V_i \cup V_j)\setminus\{u,v\}$.
\end{definition}

The condition $|V_i| \neq 1$ or $|V_j| \neq 1$ is to guarantee that $(V_i \cup V_j)\setminus\{u,v\} \neq \emptyset$, because we forbid empty part in our definition of multipartite graph.

\begin{lemma}\label{lemma:contraction}
Let $G$ be a multipartite graph.
Then $T \in \ITof{G}$ and $e \in T$
if and only if
$T\setminus \{e\} \in \ITof{G_{/e}}$.
\end{lemma}
\begin{proof}
Let $e=(u,v)$ with $u\in V_i$ and $v\in V_j$.

Assume first that $T \in \ITof{G}$ and $e \in T$.
Since $T$ is a matching, $T\setminus\{e\}$ is still a matching.
Moreover, no edge of $T\setminus\{e\}$ has one endpoint in $V_i$
and the other in $V_j$, hence it is also a matching of $G_{/e}$.

The graph $\partition(T)$ is a spanning tree of $G_{\partition}$
and contains the edge $(i,j)=\partition(e)$.
Therefore, $\partition(T)\setminus\{\partition(e)\}$
is a spanning tree of $(G_{/e})_{\partition}$.
Hence, $T\setminus\{e\}\in\ITof{G_{/e}}$.

Conversely, assume that $T\setminus\{e\} \in \ITof{G_{/e}}$.
Then $T\setminus\{e\}$ is a matching of $G_{/e}$,
and therefore also a matching of $G$,
since the contraction only affects the parts $V_i$ and $V_j$.
Adding the edge $e$ preserves the matching property,
because the vertices $u$ and $v$ do not appear in $G_{/e}$.

Moreover,
$\partition(T\setminus\{e\})$
is a spanning tree of
$(G_{/e})_{\partition}=(G_{\partition})_{/(i,j)}$.
Adding back the edge $e$,
whose projection is $(i,j)$,
yields that $\partition(T)$ is a spanning tree of $G_{\partition}$.
Thus, $T\in\ITof{G}$ and contains $e$.
\end{proof}

\begin{lemma}\label{lemma:complete}
    Let $G =( V = V_1\sqcup\dots\sqcup V_k ,E)$ a complete multipartite graph, then $\ITof{G} \neq \emptyset$ if and only if $|V| \geq 2(k-1)$.
\end{lemma}
\begin{proof}
Assume that $T \in \ITof{G}$ is an interconnection tree of $G$. Since there are $k$ parts in $G$ and $\partition(T)$ is a spanning tree of $G_{\partition}$, $T$ has $k-1$ edges. Since $T$ is a matching, it contains $2(k-1)$ vertices and thus $|V| \geq 2(k-1)$.

We prove by induction on $k$ that a complete multipartite graph with $|V| \geq 2(k-1)$ has an interconnection tree. If there is a single part, $T = \emptyset$ is an interconnection tree of $G$. If $G$ has exactly two parts $V_1$ and $V_2$, we have $v_1 \in V_1$ and $v_2 \in V_2$ because the parts are non-empty and $\{(v_1,v_2)\}$ is an interconnection tree of $G$.

Let us assume the property for $k-1 \geq 2$ parts and let us consider $G$ with $k$ parts such that $|V| \geq 2(k-1)$. W.l.o.g., let the largest part of $V$ be $V_1$. Because $k\geq 3$ and $|V| \geq 2(k-1)$, we have $|V_1|\geq 2$. We let $u \in V_1$ and $v \in V_2$ (which exists because the parts are non-empty).
The multipartite graph $G_{/(u,v)}$ has $k-1$ non-empty parts, and it satisfies $|V(G_{/(u,v)})| = |V(G)| - 2 \geq 2(k-2)$. By the induction hypothesis used on $G_{/(u,v)}$, we have an interconnection tree $T$ of $G_{/(u,v)}$. By Lemma~\ref{lemma:contraction}, $T \cup \{(v_1,v_2)\}$ is an interconnection tree of $G$, which proves the induction.
\end{proof}

As a direct corollary of \cref{lemma:complete}, we have a linear time algorithm to decide \ITP\ on complete multipartite graph.

\begin{corollary}\label{corollary:complete}
    \ITP\ can be solved in time $O(n)$ when restricted to complete multipartite graphs with $n$ vertices.
\end{corollary}

It is also possible to derive a formula for counting the number of interconnection trees of a multipartite graph using classical generating function techniques. This yields a linear-time algorithm for solving \CITP\ on complete multipartite graph.

\begin{theorem}\label{th:counting}
Let $G=(V=V_1\sqcup\dots\sqcup V_k,E)$ be a complete multipartite graph.
Then
\[
|\ITof{G}| = (k-2)!\binom{|V|-k}{k-2}\prod_{i=1}^k |V_i|.
\]
\end{theorem}
\begin{proof}
Let $S$ be a spanning tree of $G_{\partition}$, and let $d_i(S)$
denote the degree of part $i$ in $S$.
For any $T\in\ITof{G}$ such that $\partition(T)=S$,
exactly $d_i(S)$ vertices of $V_i$ are incident to edges of $T$.
Hence, the number of such trees is
\(
\prod_{i=1}^k (|V_i|)_{d_i(S)},
\)
where $(n)_m = n(n-1)\cdots(n-m+1)=\frac{n!}{(n-m)!}$
denotes the falling factorial.

Let $T_k$ be the set of labeled trees on $k$ vertices. Then
\[
|\ITof{G}|
= \sum_{S\in T_k} \prod_{i=1}^k (|V_i|)_{d_i(S)}.
\]

Set $a_i=d_i(S)-1$. Since $\sum_{i=1}^k d_i(S)=2(k-1)$,
we have $a_1+\cdots+a_k=k-2$.
The number of trees in $T_k$ with degree sequence $(d_1,\dots,d_k)$ is
\(
\frac{(k-2)!}{a_1!\cdots a_k!},
\)
a classical consequence of the bijection between labeled trees of $K_k$
and Prüfer sequences of length $k-2$ \cite{Aigner2018,DBLP:books/others/BondyM76}.

Therefore,
\[
|\ITof{G}|
=\sum_{a_1+\cdots+a_k=k-2}
\frac{(k-2)!}{a_1!\cdots a_k!}
\prod_{i=1}^k (|V_i|)_{a_i+1}.
\]

We now use classical methods from analytic combinatorics (see \cite{flajolet2009analytic}). For each $i$, define
\[
F_i(x)=\sum_{a\ge 0} (|V_i|)_{a+1}\frac{x^a}{a!}.
\]

Let $F(x)=\prod_{i=1}^k F_i(x)$. In the product expansion of $F(x)$, the coefficient of $x^{k-2}$ is
\[
\sum_{a_1+\cdots+a_k=k-2}
\frac{1}{a_1!\cdots a_k!}
\prod_{i=1}^k (|V_i|)_{a_i+1}.
\]
Hence,
\(
|\ITof{G}|=(k-2)!\,[x^{k-2}]F(x).
\)

Using the identity
\(
\sum_{m\ge 0}(n)_m\frac{x^m}{m!}=(1+x)^n
\)
and differentiating after an index shift, we obtain
\(
F_i(x)=|V_i|(1+x)^{|V_i|-1}.
\)
Therefore,
\[
F(x)=\left(\prod_{i=1}^k |V_i|\right)(1+x)^{\sum_{i=1}^k(|V_i|-1)}
=\left(\prod_{i=1}^k |V_i|\right)(1+x)^{|V|-k}.
\]

Since the coefficient of $x^{k-2}$ in $(1+x)^{|V|-k}$ is $
\binom{|V|-k}{k-2}$, we obtain
\[
|\ITof{G}|=(k-2)!\binom{|V|-k}{k-2}\prod_{i=1}^k |V_i|.
\]
\end{proof}

\subsection{Quasi complete multipartite graph}

We now consider the more general case of a \emph{quasi-complete} multipartite graph, which will be useful for the enumeration procedure in \cref{WICTEnum}.
We introduce the bipartite graph $G_{bp}$, which describes how the distinguished part $V_1$ is connected to the other parts: a vertex $u\in V_1$ is adjacent to the vertex $V_i$ in $G_{bp}$ if and only if $u$ has at least one neighbor in $V_i$.

\begin{definition}
Let $G=(V=V_1\sqcup\dots\sqcup V_k,E)$ be a quasi-complete multipartite graph.
We define $G_{bp}=(V_1\sqcup\{V_2,\dots,V_k\},E_{bp})$ as the bipartite representation of $G$, where $
E_{bp}
    = \bigl\{ (u, V_i) \,\big|\, 
        u \in V_1,\ i \ge 2,\ 
        \exists v \in V_i : (u,v) \in E
      \bigr\}$.
\end{definition}

We will need to contract several edges and to extend \cref{lemma:contraction} to this setting.
For $G$ a multipartite graph and $M$ a matching of $G$ such that $\partition(M)$ is acyclic, we define $G_{/M}$ as the contraction of $G$ along the edges of $M$ as defined in \cref{definition:contraction}.
The acyclicity of $\partition(M)$ guarantees that $M$ remains a matching after successive contractions and that the order of contractions does not affect the graph $G_{/M}$.

\begin{lemma}\label{lemma:contractionOfM}
Let $G=(V=V_1\sqcup\dots\sqcup V_k,E)$ be a multipartite graph, and let $M\subset E$ be a matching such that $\partition(M)$ is acyclic. Then $M \subseteq T$ and $T\in \ITof{G}$ if and only if $T\setminus M \in \ITof{G_{/M}}$.
\end{lemma}
\begin{proof}
By induction on $|M|$, using \cref{lemma:contraction}. The acyclicity of $\partition(M)$ ensures that $M$ remains a matching after each contraction.
\end{proof}

We now characterize the quasi-complete multipartite graphs that admit an interconnection tree.
We let $\mathrm{MaxMatching}(G)$ denote the maximum size of a matching of $G$.

\begin{lemma}\label{lemma:quasicomplete}
Let $G =( V = V_1\sqcup\dots\sqcup V_k ,E)$ be a quasi-complete multipartite graph. Then there exists an interconnection tree if and only if the following two conditions hold:
\begin{enumerate}
     \item $\left|V\right| - \left|V_1\right| \geq 
     2(k-1) - \mathrm{MaxMatching}(G_{bp})$;
     \item either for all $i\in[2..k]$, $|V_i| = 1$,
    or there exists $i \in [2..k]$ such that $|V_i|\geq 2$ and there is $(u,v) \in E(G)$ with $u\in V_1$ and $v \in V_i$.
\end{enumerate}
\end{lemma}

\begin{proof}
($\Rightarrow$) Let $T \in \ITof{G}$. The tree $T$ has exactly $k-1$ edges. Let $M$ be the subset of edges of $T$ that are incident to $V_1$ and to some part $V_i$ with $i\neq 1$. Since $T$ is a spanning tree of $G_{\partition}$, no two edges of $T$ have the same projection. Hence, $M$ is a matching in the bipartite graph $G_{bp}$.

The edges in $T\setminus M$ are incident only to parts $V_i$ with $i \neq 1$. Since $T$ is a matching, these edges are vertex-disjoint and use $2(k-1-|M|)$ distinct vertices outside $V_1$. Hence,
\[
|V| - |V_1| \ge 2(k-1) - |M|.
\]
Since $M$ is a matching of $G_{bp}$, we have $|M| \leq \mathrm{MaxMatching}(G_{bp})$, which gives the first condition.

 If all parts $V_i$ with $i\neq 1$ have cardinality one, the second condition holds. Otherwise, consider a part $V_i$ with $|V_i|\ge 2$. Since $\partition(T)$ is a spanning tree of $G_{\partition}$, there is a path connecting $V_1$ to $V_i$. Parts of cardinality one have degree one in $T$ and therefore appear only as leaves in $\partition(T)$. Hence, the second part in the path from $V_1$ to $V_i$ must also have cardinality at least two. This implies the existence of an edge between $V_1$ and some $V_j$ with $|V_j|\ge 2$, proving the second condition.

($\Leftarrow$) Assume now that both conditions hold.
Let $M$ be a maximum matching of $G_{bp}$.

If the second condition holds because all parts $V_i$ with $i\neq 1$ have size one, then $|V|-|V_1|=k-1$ and thus $|M|\ge k-1$. Since there are $k-1$ vertices on the right side of $G_{bp}$, the matching $M$ connects $V_1$ to every other part. As each $V_i$ has cardinality one, these edges also belong to $G$, and $M$ is an interconnection tree.

Otherwise, there exists $i$ and an edge $(u,v)\in E$ such that $|V_i|\ge 2$, $u\in V_1$ and $v\in V_i$. By a standard exchange argument, we may assume that $M$ contains an edge corresponding to a neighbor of $V_i$.

Let $M'$ be a lifting of $M$ to $G$, that is, for each edge $(u,V_i)\in M$, we choose an edge $(u,v)\in E$ with $v\in V_i$ and add it to $M'$.
We consider the graph $G_{/M'}$ obtained by contracting all edges of $M'$.
The graph $G_{/M'}$ has $k-|M|$ parts. In this graph, vertices originating from $V_1$ may become isolated, whereas vertices originating from other parts remain adjacent to all remaining parts.

Let $G'_{/M'}$ be the graph obtained from $G_{/M'}$ by removing isolated vertices. This graph is a complete multipartite graph, and its first part is non-empty because it contains the image of $v$.

We have $|G'_{/M'}| \ge |V| - |V_1| - |M|$.
Using the first condition, we obtain $|G'_{/M'}| \ge 2(k-1) - |M| - |M| = 2(k-|M|-1)$.
By \cref{lemma:complete}, the graph $G'_{/M'}$ admits an interconnection tree $T'$. By \cref{lemma:contractionOfM}, the set $M' \cup T'$ is an interconnection tree of $G$.
\end{proof}

\begin{corollary}\label{corollary:quasicomplete}
\ITP\ restricted to quasi-complete multipartite graphs can be solved in time $O(m\sqrt{k})$, where $m$ is the number of edges and $k$ the number of parts.
\end{corollary}
\begin{proof}
Compute the graph $G_{bp}$ and a maximum matching using the Hopcroft–Karp algorithm~\cite{DBLP:conf/focs/HopcroftK71} in time $O(|E_{bp}|\sqrt{|V_{bp}|})$. Then, using the size of this matching and Lemma~\ref{lemma:quasicomplete}, we can decide whether the graph admits an interconnection tree.
\end{proof}

Although \ITP\ on quasi-complete multipartite graphs is in $\P$, the counting problem \CITP\ is significantly harder, by a reduction from counting perfect matchings. 

\begin{theorem}
The problem \CITP\ is $\#\P$-complete when restricted to quasi-complete multipartite graphs.
\end{theorem}
\begin{proof}
We give a reduction from counting perfect matchings in bipartite graphs.
Let $G=(V_1 \sqcup V_2,E)$ be a bipartite graph with $|V_1|=|V_2|$. We construct a quasi-complete multipartite graph
$G'=(V',E')$ as follows. The vertex set $V'$ is partitioned into the part $V_1$ and, for each
vertex $v\in V_2$, a singleton part $\{v\}$. The edge set $E'$ contains all edges of $E$, together
with all edges between vertices belonging to distinct singleton parts.

In $G'$, any interconnection tree selects exactly one edge incident to each singleton part.
Hence, each interconnection tree of $G'$ corresponds bijectively to a matching of $G$ that
covers all vertices of $V_2$, that is, to a perfect matching of $G$.

Therefore, counting interconnection trees of $G'$ is equivalent to counting perfect matchings in the bipartite graph $G$. The latter problem is equivalent to computing the permanent of a $0$--$1$ matrix and is $\#\P$-complete~\cite{valiant1979complexity}.
\end{proof}

We now extend the previous result to the $t$-quasi-complete case. 
To handle $t$ non-complete parts, we rely on the fact that there is an interconnection tree in which non-complete parts can be assumed to be at bounded distance from each other in the induced spanning tree on the parts, except possibly for two of them.

\begin{theorem}
    \ITP\ restricted to $t$-quasi-complete multipartite graphs can be solved in time 
    $O((n^2k^3)^{t-1} m\sqrt{n})$, where $m$ is the number of edges, 
    $n$ the number of vertices, and $k$ the number of parts.
\end{theorem}

\begin{proof}

For $t>2$, we reduce \ITP\ to the resolution of several instances of
\ITP\ on $(t-1)$-quasi-complete multipartite graphs. The base case $t=2$ requires a separate argument, similar in spirit to the proof of \cref{corollary:quasicomplete}, but with additional care for a special configuration.

Let $G$ be a $t$-quasi-complete multipartite graph, the non-complete parts are $V_1,\dots,V_t$. Let $T \in \ITof{G}$.  In the tree $\partition(T)$, there are always two non-complete parts connected by a path $P$ whose vertices, except for the two endpoints, are only in complete parts. Moreover, we may assume that only complete parts are connected through $T$ by a path incident to some vertex inside $P$.

Let us consider such $P$, we distinguish several cases according to its length. 
If $P$ consists of up to three edges $e_1,e_2,e_3$, then these edges reduces the problem to \ITP\ on $G_{/\{e_1,e_2,e_3\}}$.

Assume now that $P$ contains more than three edges. There are again two cases. The first edge and last edge of $P$ are called respectively $e_f$ and $e_l$.

If there is a vertex $v \in V_i \cap P$, such that $|V_i|>2$, then we can rearrange the edges of $P$ but $e_f$ and $e_l$ so that we connect all parts incident to $P$ to $V_i$ except for the part connected by $e_f$ and $e_l$. This produces another interconnection tree in which there is a path $P$ with the form
$e_f,e_{in},e_{out},e_l$, where $e_{in}$ and $e_{out}$ are incident to $V_i$.
We can therefore reduce to \ITP\ on $G_{/\{e_f,e_{in},e_{out},e_l\}}$.

Otherwise, all internal parts containing a vertex inside $P$ have size exactly $2$. In this case, no other complete part can be connected to $P$ in $T$. Since $t>2$, there exists another path $P'$ between two parts among
$V_1,\dots,V_t$ with the same property: its internal parts are complete and are not connected by $T$ to non-complete parts.
This path $P'$ has length at least $4$ otherwise the problem can already be solved by the previous reduction. By rerouting the middle edges of $P$ through $P'$, we obtain a new interconnection tree in which the path $P$ has length $2$, again reducing to one of the previous cases.

We now treat the base case $t=2$. 
If $V_1$ and $V_2$ are connected by one or two edges in an interconnection
tree, we can reduce to the case $t=1$ using at most $O(kn^2)$ calls to the
algorithm of Lemma~\ref{lemma:quasicomplete}.

Assume instead that in every interconnection tree, $V_1$ and $V_2$ are at
distance at least $3$ in $\partition(T)$. If we merge $V_1$ and $V_2$ into a
single part and construct the bipartite graph $G_{bp}$ as before, then $T$
induces a matching in this graph. Compared to Lemma~\ref{lemma:quasicomplete},
we now have two additional requirements: the matching must connect both
$V_1$ and $V_2$ to the rest of the graph, and these two connections must lie
in distinct components of the matching. This leads to the same characterization
as in Lemma~\ref{lemma:quasicomplete}, except that both $V_1$ and $V_2$ must
be adjacent to a part of size at least $2$, and the inequality becomes $
\left|V\right| - \left|V_1 \cup V_2\right|
\;\ge\;
2k - \mathrm{MaxMatching}(G_{bp})$.

Finally, we bound the number of recursive calls to instances with $t-1$
arbitrary parts. To determine the path $P$, we choose one vertex as the
start of the path and another as the end, and we select at most three
intermediate parts. Any choice of edges for the path
leads to the same reduced instance, because in the interconnection tree of the proof, these only change vertices inside complete parts which are later only connected to complete parts, and all their vertices play the same role. Hence, we have fewer than $n^2k^3$ calls to the problem on $t-1$-quasi-complete graphs.
\end{proof}


We leave open for future research whether we can find an $\FPT$ algorithm for \ITP\ parametrized by $t$ for $t$-quasi-complete multigraph or if this problem is $\W[1]$-hard.


\section{Enumeration of Interconnection Trees}
\label{ICTEnum}

In this section, we turn to the problem \EITP, which motivates our study of interconnection trees.  
An enumeration problem is inherently harder than the associated decision problem, and thus we cannot hope to solve \EITP\ efficiently on arbitrary graphs. We therefore focus on complete multipartite graphs, which are relevant to our chemoinformatics applications and for which \ITP\ can be solved in linear time. The results presented here extend naturally to multipartite graphs with a bounded number of parts, as well as to $t$-quasi complete multipartite graphs.

\subsection{Binary Partition}

We first present an algorithm based on the classical \textbf{binary partition} method, which consists of splitting the solution space into two subsets: the interconnection trees that contain a given edge, and those that do not.

For an edge $e$, we denote by $G_{\setminus e}$ the multipartite graph obtained from $G$ by removing $e$.
For a set of interconnection trees $IT$ and an edge $e$, we denote by $
IT \oplus e = \{\, T \cup \{e\} \mid T \in IT \,\}$.

By \cref{lemma:contraction}, we have
$
\ITof{G}
=
\ITof{G_{\setminus e}}
\;\sqcup\;
\bigl(\ITof{G_{/e}} \oplus e\bigr)$.
Applying this decomposition recursively yields a backtracking algorithm that enumerates all interconnection trees. To make this approach efficient, we must be able to decide whether $\ITof{G_{\setminus e}}$ and $\ITof{G_{/e}}$ are empty, in order to avoid unnecessary recursive calls. Since this algorithm recurses only on subproblems that admit solutions, its delay is proportional to the depth of the recursion tree multiplied by the cost of a recursive call.

While the algorithm starts from a complete multipartite graph $G$, the graph $G_{\setminus e}$ is no longer complete, and deciding whether $\ITof{G_{\setminus e}}$ is empty is $\NP$-complete by \cref{th:ITisNP}.

A way to overcome this difficulty is to restrict the choice of the edge $e$ so that the decision problem remains tractable. Assume now that the algorithm always choose $e$ to be incident to the part $V_1$. By a simple induction, if $G$ is quasi-complete, then both $G_{\setminus e}$ and $G_{/e}$ remain quasi-complete. We can test the existence of an interconnection tree in a quasi-complete multipartite graph in time $O(m\sqrt{n})$ by \cref{corollary:quasicomplete}. Since the depth of the recursion tree is bounded by the size of an interconnection tree, the delay of the enumeration algorithm is $O(km\sqrt{n})$.

\subsection{Unbounded Branching}

Since a complete multipartite graph may contain an exponential number of interconnection trees (see \cref{th:counting}), reducing the delay is crucial.
We still rely on a flashlight search, but we design the branching choices carefully so that each recursive step reduces the problem to simpler instances of the same form, thereby ensuring optimal complexity bounds.

We refine the decomposition of $\ITof{G}$ into a finer partition.
Let $<$ be an arbitrary total order on the vertices of $G$. 
For a vertex $u$, we denote by $G_{\geq u}$ (resp. $G_{> u}$) the graph obtained from $G$ by removing all vertices $u' < u$ (resp. $u'  \le u$) such that $\partition(u)=\partition(u')$.
We decompose the space of interconnection trees by contracting all edges incident to vertices of $V_1$, in an order compatible with $<$.

\begin{lemma}\label{lemma:large_decomposition}
Let $G$ be a multipartite graph, then
\begin{equation}
\ITof{G}
=
\displaystyle
\bigsqcup_{u \in V_1}
\;\bigsqcup_{(u,v)\in E}
\ITof{(G_{\geq u}}_{/(u,v)}) \oplus (u,v).
\label{eq:large_union}
\end{equation}
If $G$ is complete multipartite, then every graph $(G_{/(u,v)})_{>u}$ is also complete multipartite.
\end{lemma}
\begin{proof}
Let $u$ be the smallest vertex of $V_1$ in the order $<$. Applying \cref{lemma:contraction} to all edges $(u,v)$ yields
\[
\ITof{G}
=
\left(
\bigsqcup_{(u,v)\in E}
\ITof{G_{/(u,v)}} \oplus (u,v)
\right)
\sqcup
\ITof{G_{\setminus \{(u,v)\in E\}}}.
\]

Since no edge incident to $u$ remains in $G_{\setminus \{(u,v)\in E\}}$, we have
\(
\ITof{G_{\setminus \{(u,v)\in E\}}} = \ITof{G_{>u}}.
\)
Hence,
\[
\ITof{G}
=
\left(
\bigsqcup_{(u,v)\in E}
\ITof{G_{/(u,v)}} \oplus (u,v)
\right)
\sqcup
\ITof{G_{>u}}.
\]

Applying this identity to $\ITof{G_{>u}}$ iteratively for each vertex of $V_1$ in increasing order, and letting $u_\ell$ be the largest vertex of $V_1$, we obtain
\[
\ITof{G}
=
\left(
\displaystyle
\bigsqcup_{u\in V_1}
\bigsqcup_{(u,v)\in E}
\ITof{(G_{\geq u})_{/(u,v)}} \oplus (u,v)
\right)
\sqcup
\ITof{G_{>u_\ell}}.
\]

Since $V_1$ is empty in $G_{>u_\ell}$, $\ITof{G_{>u_\ell}} = \emptyset$, proving \cref{eq:large_union}.

Finally, contraction and vertex removal do not remove edges between two vertices of different parts, hence $(G_{/(u,v)})_{>u}$ remains complete multipartite.
\end{proof}

\begin{theorem}
\EITP\ can be solved on complete multipartite graphs with worst-case delay $O(k)$ and amortized delay $O(1)$, where $k$ is the number of parts.
\end{theorem}

\begin{proof}
The algorithm follows the recursive structure given by \cref{eq:large_union}. For each graph that appears in this decomposition, we test in linear time whether it admits an interconnection tree using \cref{corollary:complete}, since all such graphs are complete multipartite by \cref{lemma:large_decomposition}. The algorithm then recurses on the graphs which admits an interconnection tree.

The depth of the recursion tree is at most $k$, as each recursive call decreases the number of parts by one. A naive implementation would therefore yield a delay in $O(nk)$, where $n$ is the number of vertices.

This delay can be improved by maintaining suitable data structures and by processing the unions in \cref{eq:large_union} in increasing order of $u$. The criterion of \cref{lemma:complete}, used to decide whether a complete multipartite graph admits an interconnection tree, can be updated in constant time by maintaining the current number of parts and vertices.

Indeed, when considering a graph of the form $(G_{/(u,v)})_{>u}$, the number of parts decreases by one due to the contraction. The number of vertices decreases by two because of the contraction, and then decreases by one each time we move to a contraction involving a new vertex $u$, since the previous vertex is removed.

Moreover, each graph $(G_{/(u,v)})_{>u}$ can be obtained from the previous one in constant time. We represent $V_1$ as a list of original parts, and each original part as a linked list of its vertices to allow constant-time insertions and deletions. For each vertex, we store a pointer to the list representing its part and its position in this list.

Transitioning from $(G_{/(u,v')})_{>u}$ to $(G_{/(u,v)})_{>u}$ consists of restoring $v'$ to its part, removing $v$, and updating the representation of $V_1$ accordingly. Similar constant-time updates apply when moving from $(G_{/(u,v)})_{>u}$ to $(G_{/(u',v')})_{>u'}$. Backtracking operations are the exact reverse and also take constant time. Hence, each recursive step costs $O(1)$, and the worst-case delay is $O(k)$.

For the amortized analysis, we assume that the algorithm always chooses $V_1$ to be the largest part. We charge the time required to compute the graphs appearing in \cref{eq:large_union} from $G$, as well as the time needed to detect the first empty set $\ITof{(G_{/(u,v)})_{>u}}$, evenly among all recursive calls that correspond to non-empty sets $\ITof{(G_{/(u,v)})_{>u}}$. By the previous analysis, this cost per recursive call is constant; we denote it by $C$.

If $|V_1|=1$ and there are more than two parts, then no interconnection tree exists. Otherwise, when $|V_1|=1$ and there are at most two parts, the unique interconnection tree is produced in constant time, therefore the case with more than two parts and $|V_1|>1$ is the only relevant for analyzing the amortized delay.

Assume now that the number of parts is $l>2$ and $|V_1|\ge 2$, then for any edge $(u,v)$ we have $|\ITof{G}| \;\ge\; 2\,\bigl|\ITof{(G_{/(u,v)})_{>u}}\bigr|.$ Hence, by a simple induction, when the algorithm is called on such a graph, it eventually produces at least $2^{\,l}$ interconnection trees.
If we charge the cost $C$ of such a recursive call evenly to all the interconnection trees it produces, each tree is charged at most $C\,2^{-l}$.

An interconnection tree can be charged by recursive calls at most once per level of the recursion tree. Summing over all levels, we obtain a total charge for a tree of at most $\sum_{l\le k-1} C\,2^{-l} \;\le\; 2C$. Therefore, the amortized delay is constant.
\end{proof}

Since this algorithm is based on flashlight search, generic techniques can transform amortized delay into worst-case delay~\cite{uno03,capelli2021amortized}. However, a delay linear in the size of the output is essentially optimal, unless one allows the algorithm to output \emph{deltas} between consecutive solutions instead of full solutions. In this setting, it may be possible to achieve constant worst-case delay, but the techniques of~\cite{uno03,capelli2021amortized} would first need to be adapted to handle this representation. Another possible improvement would be to design a Gray code for interconnection trees, in which consecutive solutions differ by only a constant number of edges, as is known for spanning trees in complete graphs~\cite{liu2026generating} and general graphs~\cite{merino2024traversing}.

\section{Weighted Enumeration of Interconnection Trees}
\label{WICTEnum}
In this section, we consider a weighted variant of interconnection trees in multipartite graphs. Given an integer weight function $w$ on the edges of a multipartite graph, the weight of an interconnection tree $T$ is defined as $w(T) = \sum_{e \in T} w(e)$.

In our application to molecular graphs, this weight estimates the total length of the molecular paths that must be constructed according to the interconnection tree. Consequently, trees of small weight are preferred, and we would like to enumerate interconnection trees in increasing order of weight.

Formally, given a weighted multipartite graph $G$, we consider the problem of listing all elements of $\ITof{G}$ in nondecreasing order of weight. This problem is computationally difficult, since even finding a minimum-weight interconnection tree is hard.

This difficulty already arises in the Euclidean graphs encountered in our applications, where each vertex has coordinates in $\mathbb{N}^3$ and the weight of an edge is the Euclidean distance between its endpoints. Indeed, the Euclidean Traveling Salesperson Problem, which is $\NP$-complete~\cite{DBLP:journals/tcs/Papadimitriou77}, reduces to finding a minimum-weight interconnection tree using the construction of the proof of \cref{th:ITisNP}.

\begin{theorem}
Finding a minimum-weight interconnection tree in a Euclidean complete multipartite graph is $\NP$-hard.
\end{theorem}

Although exact enumeration in nondecreasing order of weight is intractable in general,
we propose a heuristic that tends to output low-weight trees early, without guaranteeing
a strict ordering. 

We rely on a variant of Equation~\ref{eq:large_union}. Let $G$ be a weighted
quasi-multipartite graph. For a value $val$, we denote by $G_{> val}$ the graph obtained
from $G$ by removing all edges incident to $V_1$ whose weight is at most $val$.
This graph remains quasi-complete multipartite.

We obtain the following decomposition:
\[
\ITof{G}
=
\displaystyle
\bigsqcup_{\substack{(u,v)\in E \\ u\in V_1}}
\ITof{(G_{/(u,v)})_{> w(u,v)}} \oplus (u,v).
\]

Since all graphs $(G_{/(u,v)})_{> w(u,v)}$ are quasi-complete, we can test in polynomial time whether $\ITof{(G_{/(u,v)})_{> w(u,v)}}=\emptyset$ using \cref{corollary:quasicomplete}. This yields a flashlight-search algorithm in which the recursive calls are processed in increasing order of the weights of the edges incident to $V_1$. We call this algorithm \emph{Weight-Guided Enumeration} (\textsc{WGE}).

The algorithm greedily selects the lightest edges first, which tends to produce interconnection trees of small total weight early in the enumeration. In particular, the trees are produced in nondecreasing order of weight in the special cases where the graph has only two parts, or when all complete parts have cardinality~$1$. 
However, there is no general guarantee on the order: it is easy to construct instances where an interconnection tree output early by the algorithm has arbitrarily larger weight than a minimum-weight one which is produced at the end of the enumeration.

A key difference between \textsc{WGE} and the algorithm of \cref{ICTEnum}, is that it removes all edges incident to $V_1$ below a given threshold of weight. As a result, the recursive calls are on quasi-complete graphs rather than complete graphs. Consequently, the delay increases to $O(km\sqrt{n})$ instead of $O(k)$.

We may also consider alternative weight measures, such as
\[
w_{\max}(T) = \max_{e \in T} w(e)
\qquad\text{or}\qquad
w_2(T) = \sqrt{\sum_{e \in T} w(e)^2},
\]
which are also relevant in our application. Using $w$ or $w_2$ does not affect the
complexity of the problems of finding the smallest interconnection tree, since they induce the same ordering on trees up to squaring the weights.

Finding a minimum interconnection tree with respect to $w_{\max}$ is also hard. Indeed,
we can simulate \ITP\ on arbitrary graphs, which is $\NP$-complete by
Theorem~\ref{th:ITisNP}. Given a graph $G$, construct a complete multipartite graph $G'$
with the same vertices and parts, assign weight $0$ to the edges of $G$, and weight $1$
to all other edges. Then $G$ has an interconnection tree if and only if $G'$ has an
interconnection tree of $w_{\max}$-weight $0$.

\subsection{Experimental Evaluation}

We experimentally compare \textsc{WGE} to the algorithm of \cref{ICTEnum} without any weight-aware ordering that we call \emph{Baseline Enumeration} (\textsc{BE}).
To assess the quality of the output order, we compute two classical measures of sortedness~\cite{DBLP:journals/csur/Estivill-CastroW92}: the number of \emph{inversions} (pairs of solutions enumerated in the wrong order) and the number of \emph{increasing runs} (maximal sequences of consecutively enumerated trees with nondecreasing weight).

We consider both synthetic and real-world instances. The synthetic instance (SYNTHETIC) consists of $18$ vertices partitioned into six parts $(5,4,3,3,2,1)$; we generate $50$ variants with fixed topology, and in each variant edge weights are sampled uniformly between $0.1$ and $10.0$. The real-world graphs are derived from binding motifs~\cite{demange2026computational}, based on structures from the Cambridge Structural Database (CSD)~\cite{CSD}. Vertices correspond to atoms with 3D coordinates and edge weights are Euclidean distances. The instances range from $4$ to $13$ parts and from $13$ to $33$ vertices, with $2$ to $7$ vertices per part. The source code and datasets used in this work are available in our GitHub repository\footnote{\url{https://github.com/NoeDemange/InterconnectionTreesEnumeration}}.

Before comparing \textsc{WGE} with the baseline enumeration strategy, we first investigate whether its performance depends on the heuristic used to select the main part at the beginning of the enumeration process.
We evaluate three heuristics:
\begin{itemize}
    \item \emph{MaxV}, selecting the part with the largest number of vertices.
    \item \emph{MinAvg}, selecting the part minimizing the average incident edge weight.
    \item \emph{MinEdge}, selecting the part incident to the currently lightest edge.
\end{itemize}

To compare these heuristics, we measure both the quality of the enumeration order and the quality of the early generated solutions.
For the ordering quality, we use two normalized measures derived from classical sortedness metrics:
the normalized inversion ratio
$\frac{2\,\#\mathrm{inversion}}{n(n-1)}$,
and the normalized run ratio
$\frac{\#\mathrm{runs}-1}{n-1}$,
where $n$ denotes the number of enumerated solutions.
Both measures belong to $[0,1]$, and lower values indicate an ordering closer to a nondecreasing enumeration by weight.

Table~\ref{tab:part-selection} reports the results.
The left table presents the normalized ordering measures, while the right table reports the mean weight of the first $10\,000$ generated interconnection trees.
Overall, the three heuristics exhibit relatively similar behavior.
However, \emph{MinEdge} generally achieves slightly better ordering measures and tends to produce lower-weight solutions earlier in the enumeration.
For this reason, we use \emph{MinEdge} in the remainder of the experiments.

\begin{table}[ht]
\centering
\setlength{\tabcolsep}{2pt}
\scriptsize
\caption{Impact of part-selection heuristics on enumeration quality.
Left: normalized inversion (NI) and run (NR) measures.
Right: mean weight of the first $10\,000$ generated interconnection trees.}
\label{tab:part-selection}

\begin{subtable}[t]{0.49\textwidth}
\vspace{0pt}
\centering

\begin{tabular}{lcc|cc|cc}
 &
\multicolumn{2}{c|}{\textbf{MaxV}} &
\multicolumn{2}{c|}{\textbf{MinAvg}} &
\multicolumn{2}{c}{\textbf{MinEdge}} \\

\textbf{Instance} &
\textbf{NI} &
\textbf{NR} &
\textbf{NI} &
\textbf{NR} &
\textbf{NI} &
\textbf{NR} \\
\hline
\T
ABABEL   & 0.32 & 0.04 & 0.32 & 0.04 & 0.32 & 0.04 \\
ACANIL01 & 0.30 & 0.09 & 0.31 & 0.09 & 0.28 & 0.06 \\
ADENOS10 & ---  & ---  & ---  & ---  & ---  & ---  \\
ALFUCO   & 0.27 & 0.18 & 0.36 & 0.15 & 0.27 & 0.18 \\
CUGDIR   & 0.34 & 0.26 & 0.34 & 0.26 & 0.34 & 0.26 \\
DIGOXN10 & ---  & ---  & ---  & ---  & ---  & ---  \\
FAZHET01 & 0.32 & 0.18 & 0.33 & 0.17 & 0.30 & 0.15 \\
\hline\hline
\T
SYNTHETIC & 0.32 & 0.18 & 0.34 & 0.15 & 0.33 & 0.17 \\
\hline
\end{tabular}
\end{subtable}
\hfill
\begin{subtable}[t]{0.49\textwidth}
\vspace{0pt}
\centering

\begin{tabular}{lc|c|c}

\multicolumn{4}{c}{}\\

\textbf{Instance} &
\textbf{MaxV} &
\textbf{MinAvg} &
\textbf{MinEdge} \\
\hline
\T
ABABEL   & 17.73 & 17.73 & 17.73 \\
ACANIL01 & 18.58 & 18.66 & 18.53 \\
ADENOS10 & 35.63 & 35.63 & 34.25 \\
ALFUCO   & 24.93 & 25.72 & 24.93 \\
CUGDIR   & 61.19 & 61.19 & 61.19 \\
DIGOXN10 & 91.61 & 91.61 & 91.61 \\
FAZHET01 & 27.91 & 28.99 & 27.91 \\
\hline\hline
\T
SYNTHETIC & 15.53 & 15.85 & 15.51 \\
\hline
\end{tabular}

\end{subtable}

\end{table}

We now compare the \textsc{WGE} approach, with the baseline enumeration strategy (\textsc{BE}) that does not incorporate any weight-aware ordering.
Table~\ref{tab:combined} summarizes the results. On the left, we report the ratios (\textsc{WGE}/\textsc{BE}) for inversions and runs: values below $1$ indicate an improvement of the weight-guided strategy. Missing values correspond to instances for which full enumeration is not tractable. On the right, we report the mean weight of the first $10\,000$ enumerated solutions for both methods. This reflects the practical scenario where only the first solutions are generated, making the quality of early outputs particularly important.

Overall, the weight-guided strategy consistently improves the ordering quality and yields lower-weight solutions earlier in the enumeration.

\begin{table}[ht]
\centering
\scriptsize
\caption{Impact of ordering on enumeration quality: ratios of inversions and runs (left) and mean weights for the first $10\,000$ interconnection trees produced by the two algorithms (right).}
\label{tab:combined}

\begin{subtable}[t]{0.48\textwidth}
\vspace{0pt}
\centering
\begin{tabular}{lll}
\textbf{Instance} &
\textbf{Inv. ratio} &
\textbf{Run ratio} \\
\hline
\T
ABABEL   & 0.706 & 0.075 \\
ACANIL01 & 0.636 & 0.112 \\
ADENOS10 & ---   & ---   \\
ALFUCO   & 0.656 & 0.353 \\
CUGDIR   & 0.771 & 0.550 \\
DIGOXN10 & ---   & ---   \\
FAZHET01 & 0.602 & 0.282 \\
\hline\hline
\T
SYNTHETIC & 0.654 & 0.343 \\
\hline
\end{tabular}
\end{subtable}
\hfill
\begin{subtable}[t]{0.48\textwidth}
\vspace{0pt}
\centering
\begin{tabular}{lcc}
\textbf{Instance} &
\textbf{\textsc{BE}} &
\textbf{\textsc{WGE}} \\
\hline
\T
ABABEL    & 25.78 & 17.73 \\
ACANIL01  & 19.49 & 18.53 \\
ADENOS10  & 52.37 & 34.25 \\
ALFUCO    & 26.47 & 24.93 \\
CUGDIR    & 69.79 & 61.19 \\
DIGOXN10  & 139.42 & 91.61 \\
FAZHET01  & 37.50 & 27.91 \\
\hline\hline
\T
SYNTHETIC & 25.36 & 15.52 \\
\hline
\end{tabular}
\end{subtable}

\end{table}

\newpage 
\bibliography{lipics-v2021-article}

\end{document}